\documentclass[a4paper,10pt]{article}
\usepackage{amsmath}
\usepackage{amsthm}
\usepackage{amssymb}
\usepackage{etex,mathtools}

\usepackage{nicefrac}
\usepackage{tikz} 
\usepackage{times} 
\usepackage{wrapfig}

\usepackage{url,xspace}
\usepackage{subfig} 

\usepackage[all]{xy}
\newcommand{\Z}{\mathbb{Z}}
\newcommand{\customvspace}[1]{}

\theoremstyle{plain}

\newtheorem{theorem}{Theorem}
\newtheorem{lemma}{Lemma}

\usepackage{enumitem}
\setlist{topsep=2pt}

\usepackage{color,svgcolor,hyperref}
\hypersetup{breaklinks=true, 
  colorlinks=true,
  linkcolor=blue,
  citecolor=darkred,
  urlcolor=darkblue,
}

\begin{document}
\title{Brandt's Fully Private Auction Protocol Revisited}
\author{Jannik Dreier\thanks{Universit\'e Grenoble 1, CNRS, Verimag,
    France. \href{mailto:Jannik.Dreier@imag.fr,Pascal.Lafourcade@imag.fr}{\{Jannik.Dreier,Pascal.Lafourcade\}@imag.fr}} \and Jean-Guillaume Dumas\thanks{Universit\'e Grenoble 1,
    CNRS,  Laboratoire Jean Kuntzmann (LJK), France. \href{mailto:Jean-Guillaume.Dumas@imag.fr}{Jean-Guillaume.Dumas@imag.fr}} \and  Pascal
  Lafourcade\footnotemark[1]}

\maketitle

\begin{abstract}
  Auctions have a long history, having been recorded as early as 500
  B.C.~\cite{Krishna02}. Nowadays, electronic auctions have been a
  great success and are increasingly used. Many cryptographic
  protocols have been proposed to address the various security
  requirements of these electronic transactions, in particular to
  ensure privacy. Brandt~\cite{Brandt06} developed a protocol
  that computes the winner using homomorphic operations on a
  distributed ElGamal encryption of the bids. He claimed that it
  ensures full privacy of the bidders, i.e. no information apart
  from the winner and the winning price is leaked. We first show that
  this protocol -- when using malleable interactive zero-knowledge
  proofs -- is vulnerable to attacks by dishonest bidders. Such
  bidders can manipulate the publicly available data in a way that
  allows the seller to deduce all participants' bids. 
  Additionally we discuss some issues with verifiability as well as
  attacks on non-repudiation, fairness and the privacy of individual
  bidders exploiting authentication problems.

\end{abstract}

\section{Introduction}\label{sec:intro}

Auctions are a simple method to sell goods and services. Typically a
\emph{seller} offers a good or a service, and the \emph{bidders} make
offers. Depending on the type of auction, the offers might be sent
using sealed envelopes which are opened simultaneously to determine
the winner (the ``sealed-bid'' auction), or an \emph{auctioneer} could
announce prices decreasingly until one bidder is willing to pay the
announced price (the ``dutch auction''). Additionally there might be
several rounds, or offers might be announced publicly directly (the
``English'' or ``shout-out'' auction). The winner usually is the
bidder submitting the highest bid, but in some cases he might only
have to pay the second highest offer as a price (the ``second-price''-
or ``Vickrey''-Auction). In general a bidder wants to win the auction
at the lowest possible price, and the seller wants to sell his good at
the highest possible price. For more information on different auction
methods see \cite{Krishna02}.  To address this huge variety of
possible auction settings and to achieve different security and
efficiency properties numerous protocols have been developed,
e.g. \cite{Brandt06,Curtis07,Naor99,Omote01,Peng02,Sadeghi02,Sako00}
and references therein.

One of the key requirements of electronic auction (e-Auction)
protocols is privacy, i.e. the bids of losing bidders remain
private. Brandt proposed a first-price
sealed-bid auction protocol~\cite{Brandt06,Brandt03,Brandt02} and claimed that it is
fully private, i.e. it leaks no information apart from the
winner, the winning bid, and what can be deduced from these two facts
({\it e.g.} that the other bids were lower).

\paragraph{Our Contributions.}
The protocol is based on an algorithm that computes the winner using bids encoded as bit vectors.
In this paper we show that the implementation using the homomorphic property of a distributed Elgamal
encryption proposed in the original paper suffers from a weakness. 
In fact, we
prove that any two different inputs (i.e. different bids) result in
different outcome values, which are only hidden using random
values. We show how a dishonest participant can remove this random
noise, if malleable interactive zero-knowledge proofs are used. The seller can
then efficiently compute the bids of all bidders, hence completely
breaking privacy. We also discuss two problems with verifiability, and how the lack of authentication enables attacks on privacy
even if the above attack is prevented via non-malleable
non-interactive proofs. Additionally we show attacks on
non-repudiation and fairness, and propose solutions to all discovered
flaws in order to recover a fully resistant protocol.

\paragraph{Outline.}
In the next section, we recall the protocol of Brandt. 
Then, in the following sections, we present our attacks in
several steps. In Section~\ref{sec:attackcomp}, we first study the
protocol using interactive zero-knowledge proofs and without noise. 
Then we show how a dishonest participant can
remove the noise, thus mount the attack on the protocol with noise,
and discuss countermeasures. Finally, in
Section~\ref{sec:attackverif}, we 
discuss verifiability and in Section~\ref{sec:attackauth} we discuss
attacks on fairness, non-repudiation and privacy exploiting the lack
of authentication.

\section{The Protocol}\label{sec:protocol}

The protocol of Brandt~\cite{Brandt06} was designed to ensure full
privacy in a completely distributed way. It exploits the homomorphic
properties of a distributed El-Gamal encryption
scheme~\cite{ElGamal85} for a secure multi-party computation of the
winner. Then it uses zero-knowledge proofs of knowledge of discrete
logarithms to ensure correctness of the bids while preserving privacy.
We first give a high level description of the protocol and
then present details on its main cryptographic primitives.

\subsection{Informal Description}

The participating $n$ bidders and the seller communicate essentially
using broadcast messages. The latter can for example be implemented
using a bulletin board, i.e. an append-only memory accessible to
everybody.  The bids are encoded as $k$-bit-vectors where each entry
corresponds to a price. If the bidder $a$ wants to bid the price
$b_a$, all entries will be $1$, except the entry $b_a$ which will be
$Y$ (a public constant). Each entry of the vector is then encrypted
separately using a $n$-out-of-$n$-encryption scheme set up by all
bidders. The bidders use multiplications of the encrypted values to
compute values $v_{a j}$, exploiting the homomorphic property of the
encryption scheme. Each
one of this values is $1$ if the bidder $a$ wins at price $j$, and is
a random number otherwise.  The decryption of the final values takes
place in a distributed way to ensure that nobody can access
intermediate values.

\subsection{Mathematical Description  (Brandt~\cite{Brandt06})}

Let $\mathbb{G}_q$ be a multiplicative subgroup of order $q$, prime,
and $g$ a generator of the group.
We consider that $i, h \in \{1,
\ldots, n\}$, $j, bid_a \in \{1, \ldots, k\}$ (where $bid_a$ is the bid chosen by the bidder with index $a$), $Y \in
\mathbb{G}_q\setminus{\{1\}}$.  More precisely, the $n$ bidders execute
the following five steps of the protocol~\cite{Brandt06}:
\begin{enumerate}
\item {\bf Key Generation}\\
Each bidder $a$, whose bidding price is $bid_a$ among $\{1, \ldots, k\}$ does the following:
\begin{itemize}
 \item chooses a secret $x_a \in \Z/q\Z$
 \item chooses randomly $m^{a}_{i j}$ and $r_{a
 j} \in \Z/q\Z$ for each $i$ and $j$. 
 \item publishes $y_{a} = g^{x_a}$ and proves the knowledge of $y_a$'s discrete logarithm.
 \item using the  published $y_i$ then computes $y = \prod_{i=1}^{n} y_i$.
\end{itemize}
\item\label{it:bid} {\bf Bid Encryption}\\
Each bidder $a$ 
\begin{itemize}
 \item sets $b_{a j} = \begin{cases}
   Y & \mbox{ if } j = bid_a \\
   1 & \mbox{ otherwise}
  \end{cases}$
  \item publishes $\alpha_{a j} = b_{a j} \cdot y^{r_{a j}}$ and $\beta_{a j} = g^{r_{a j}}$ for each $j$.
 \item proves that for all $j$, $\log_g(\beta_{a j})$ equals
   $\log_y(\alpha_{a j})$ or $\log_y\left(\frac{\alpha_{a j}}{Y}\right)$, and that \\ 
   $ \log_y\left(\frac{\prod_{j=1}^k \alpha_{a j}}{Y}\right) = \log_g\left(\prod_{j=1}^k \beta_{a j}\right)$.
\end{itemize}
\item\label{it:outcome} {\bf Outcome Computation}
\begin{itemize}
 \item Each bidder $a$ computes and publishes for all $i$ and $j$:
\begin{equation*}
 \gamma_{i j}^{a} = \left( \left( \prod_{h=1}^n \prod_{d=j+1}^k \alpha_{h d} \right) \cdot \left( \prod_{d=1}^{j-1} \alpha_{i d} \right) \cdot \left( \prod_{h=1}^{i-1} \alpha_{h j} \right) \right)^{m_{i j}^a}
\end{equation*}
\begin{equation*}
 \delta_{i j}^{a} = \left( \left( \prod_{h=1}^n \prod_{d=j+1}^k \beta_{h d} \right) \cdot \left( \prod_{d=1}^{j-1} \beta_{i d} \right) \cdot \left( \prod_{h=1}^{i-1} \beta_{h j} \right) \right)^{m_{i j}^a}
\end{equation*}
\\and proves its correctness.
\end{itemize}
\item {\bf  Outcome Decryption}
\begin{itemize}
 \item Each bidder $a$ sends $\phi_{i j}^a = (\prod_{h=1}^n \delta_{i j}^h)^{x_a}$ for each $i$ and $j$ to the seller and proves its correctness. After having received all values, the seller publishes $\phi_{i j}^h$ for all $i$, $j$, and $h \neq i$.
\end{itemize}
\item {\bf Winner determination}
\begin{itemize}
 \item Everybody can now compute $v_{a j} = \frac{\prod_{i=1}^n \gamma_{a j}^i}{\prod_{i=1}^n \phi_{a j}^i}$ for each $j$.
 \item If $v_{a w} = 1$ for some $w$, then the bidder $a$ wins the auction at price $p_w$.
\end{itemize}
\end{enumerate}

\subsection{Malleable proofs of knowledge and discrete
  logarithms}\label{ssec:PoK}
\newcommand{\pdl}{\hyperlink{def:pdl}{PDL}\xspace}
\newcommand{\eqdl}{\hyperlink{def:eqdl}{EQDL}\xspace}

In the original paper~\cite{Brandt06} the author suggests using zero-knowledge proofs of knowledge to protect against active adversaries. The basic protocols he proposes are interactive and malleable, but can be converted into non-interactive proofs using the Fiat-Shamir heuristic~\cite{Fiat86}, as advised by the author. 
We first recall the general idea of such proofs, then we expose the man-in-the-middle
attacks on the interactive version, which we will use as part of our first attack.

Let \pdl denote a {\em \hypertarget{def:pdl}{proof of knowledge of a discrete logarithm}}. A
first scheme for \pdl was developed in 1986 by Chaum et al.~\cite{Chaum86}. In the original auction
paper~\cite{Brandt06} Brandt proposes to use a non-interactive variant of \pdl as developed by
Schnorr~\cite{Schnorr91}, which are malleable. Unfortunately, interactive malleable \pdl are subject to
man-in-the-middle attacks~\cite{Katz2002}.  We first recall the classic
$\Sigma$-protocol on a group with generator $g$ and order $q$
\cite{Bangerter05,Burmester89,Chaum87}.  Peggy and Victor know $v$ and $g$, but
only Peggy knows $x$, so that $v = g^x$.  She can prove this fact, without
revealing $x$, by executing the following protocol:
\begin{enumerate}
\item Peggy chooses $r$ at random and sends $z = g^r$ to Victor.
\item Victor chooses a challenge $c$ at random and sends it to Peggy.
\item Peggy sends $s = (r + c\cdot x) \mod q$ to Victor.
\item Victor checks that $g^s= z\cdot v^c$.
\end{enumerate}

\subsubsection[Man-in-the-middle attacks on interactive PDL]{Man-in-the-middle attacks on interactive \pdl}\label{sec:MITMPDL}
Suppose Peggy possesses some secret discrete logarithm $x$.  We present here the
man-in-the-middle attack of~\cite{Katz2002}, where an attacker can pretend to
have knowledge of any affine combination of the secret $x$, even providing the
associated proof of knowledge, without breaking the discrete logarithm.  To
prove this possession to say Victor, the attacker will start an interactive proof
knowledge session with Peggy and another one with Victor. The attacker will
transform Peggy's outputs and forward Victor's challenges to her. 
The idea is to use the proof
of possession of Peggy's $x$, to prove possession of $1-x$ to Victor. Indeed to
prove for instance possession of just $x$ to Victor, an attacker would only have to
forward Peggy's messages to Victor and Victor's messages to Peggy. The idea of the
attack is similar, except that one needs to modify the messages of Peggy. We
show the example of $1-x$ in Figure~\ref{fig:exm} since it is used in 
Section~\ref{ssec:noiserem} to mount our attack. Upon demand by Victor to prove knowledge of $1-x$, Mallory,
the man-in-the-middle, simply starts a proof of knowledge of $x$ with
Peggy. Peggy chooses a random exponent $r$ and sends the commitment $z=g^r$ to
Mallory. Mallory simply inverts $z$ and sends $y=z^{-1}$ to Victor. Then Victor
presents a challenge $c$ that Mallory simply forwards without modification to Peggy. Finally
Peggy sends a response $s$ that Mallory combines with $c$, as $u=c-s$, to
provide a correct answer to Victor. This is summarized in Figure~\ref{fig:exm}.
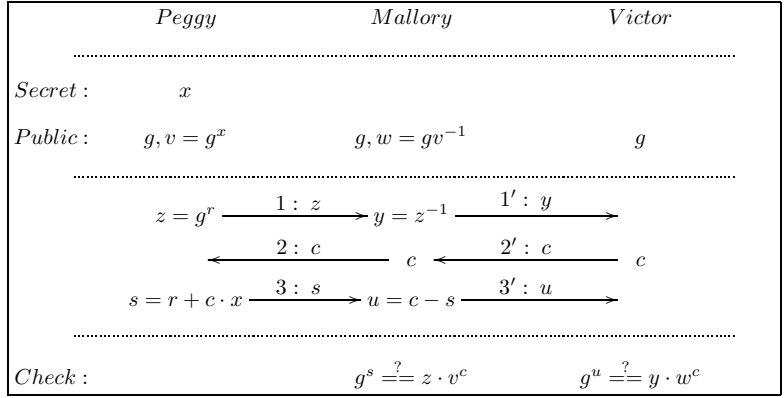
\begin{figure}[ht]\center\customvspace{-15pt}
\noindent\resizebox{.85\linewidth}{!}{$$
\xy
\xymatrix@R=7pt@C=1pc@W=15pt{
&Peggy && Mallory && Victor&\\
\ar@{.}[rrrrrr]&&&&&&\\
Secret:&x&&&&\\
Public:&g,v=g^x&&g,w=gv^{-1}&&g\\
\ar@{.}[rrrrrr]&&&&&&\\
&z=g^r\ar[rr]^*[*1.]{1:\ z} && y=z^{-1}\ar[rr]^*[*1.]{1':\ y}&&\\
&&& c\ar[ll]_*[*1.]{2:\ c} && c\ar[ll]_*[*1.]{2':\ c}\\
&s=r+c\cdot x\ar[rr]^*[*1.]{3:\ s} && u=c-s\ar[rr]^*[*1.]{3':\ u}&&\\
\ar@{.}[rrrrrr]&&&&&&\\
Check:&&& g^s \stackrel{?}{==} z\cdot v^c && g^u \stackrel{?}{==} y\cdot w^c\\
}\POS*\frm{-}
\endxy
$$
}
\caption{Man-in-the-middle \protect\pdl of $1-x$, with $x$ an unknown discrete logarithm.}\label{fig:exm}\customvspace{-15pt}
\end{figure}

Actually, the attack works in the generic settings
of~\cite{Burmester89,Maurer09} or of $\Sigma$-protocols~\cite{Cramer98}. We let
$f:\Gamma\rightarrow\Omega$ denote a one way homomorphic function between two
commutative groups $(\Gamma,+)$ and $(\Omega,\times)$.   
We use this generalization to prevent possible countermeasures of our
first attack in Section~\ref{ssec:countermitm}.

For an integral value $\alpha$, $\alpha \cdot x \in \Gamma$
(resp. $y^\alpha\in\Omega$) denotes $\alpha$ applications of the group law $+$
(resp. $\times$). 
For a secret $x\in\Gamma$, and any $(h,\alpha,\beta)\in\Gamma\times\Z^2$, the
attacker can build a proof of possession of $\alpha\cdot h+\beta\cdot x$.
In the setting of the example of Figure~\ref{fig:exm}, we used $f(x)=g^x$,
$h=1$, $\alpha=1$ and $\beta=-1$. 

In the general case also, upon demand of proof by Victor, Mallory starts a proof
with Peggy. The secret of Peggy is $x$, and the associated witness $v$ is
$v=f(x)$. Then Mallory wants to prove that his witness $w$ corresponds to any
combination of $x$ with a logarithm $h$ that he knows. With only public
knowledge and his chosen $(h,\alpha,\beta)\in\Gamma\times\Z^2$, Mallory is able
to compute $w=f(h)^\alpha\cdot v^\beta$.

For the proof of knowledge, Mallory still modifies the commitment $z=f(r)$
of Peggy to $y=z^\beta$. Mallory forwards the challenge $c$ of Victor without modification.
Finally Mallory transforms the response $s$ of Peggy, still with only
public knowledge and his chosen $(h,\alpha,\beta)\in\Gamma\times\Z^2$,
as $u=c\cdot(\alpha\cdot h)+\beta\cdot s$. 
We summarize this general attack on Figure~\ref{fig:generic}. 
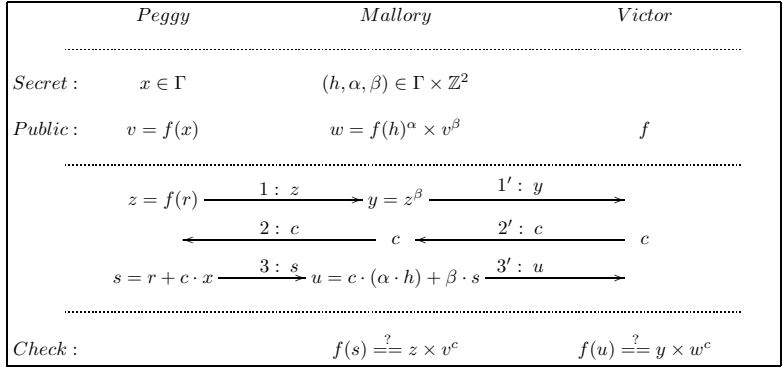
\begin{figure}[ht]\center\customvspace{-15pt}
\noindent\resizebox{.85\linewidth}{!}{$$
\xy
\xymatrix@R=7pt@C=1pc@W=15pt{
&Peggy && Mallory && Victor&\\
\ar@{.}[rrrrrr]&&&&&&\\
Secret:&x\in\Gamma&& (h,\alpha,\beta)\in\Gamma\times\Z^2 &&\\
Public:&v=f(x)&&w=f(h)^\alpha \times v^\beta && f\\
\ar@{.}[rrrrrr]&&&&&&\\
&z=f(r)\ar[rr]^*[*1.]{1:\ z} && y=z^\beta\ar[rr]^*[*1.]{1':\ y}&&\\
&&& c\ar[ll]_*[*1.]{2:\ c} && c\ar[ll]_*[*1.]{2':\ c}\\
&s=r+c\cdot x\ar[rr]^*[*1.]{3:\ s} && u=c\cdot(\alpha\cdot h)+\beta\cdot s\ar[rr]^*[*1.]{3':\ u}&&\\
\ar@{.}[rrrrrr]&&&&&&\\
Check:&&&f(s) \stackrel{?}{==} z\times v^c && f(u) \stackrel{?}{==} y\times w^c\\
}\POS*\frm{-}
\endxy
$$
}
\caption{Man-in-the-middle attacks proving knowledge of affine transforms of a secret discrete logarithm in the generic setting.}\label{fig:generic}
\customvspace{-15pt}\end{figure}

\begin{lemma} In the man-in-the-middle attack of Figure~\ref{fig:generic} of the
  interactive proof of knowledge of a discrete logarithm, Victor is convinced by
  Mallory's proof of knowledge of $\alpha\cdot h+\beta\cdot x$. 
\end{lemma}
\begin{proof}
Indeed, 
\begin{equation}\label{eq:mitmresp}
u=c\cdot(\alpha\cdot h)+\beta\cdot s=c\cdot(\alpha\cdot h)+\beta\cdot(r+c\cdot x)=\beta\cdot
r+c\cdot(\alpha\cdot h+\beta\cdot x).
\end{equation}
Now, since $z=f(r)$, $y=z^\beta$, $v=f(x)$ and $f(h)^\alpha\times v^\beta=w$, the
latter Equation~(\ref{eq:mitmresp})
proves in turn that 
\begin{equation}\label{eq:mitmpdl}
f(u) = f(r)^\beta \times f(\alpha\cdot h+\beta\cdot x)^c = z^\beta \times
(f(h)^\alpha\times f(x)^\beta)^c= y\times w^c .
\end{equation}
Now Victor has to verify the commitment-challenge-response $(y,c,u)$ of Mallory for
his witness $w$. Then Victor needs to checks whether $f(u)$ corresponds to 
$y\times w^c$, which is the case as shown by the latter
Equation~(\ref{eq:mitmpdl}).\qed
\end{proof}
\subsubsection{Generalizations to equality of discrete logarithms}\label{ssec:eqdl}
We let \eqdl denote a {\em \hypertarget{def:eqdl}{proof of equality of several discrete logarithms}}.
Any \pdl can in general easily be transformed to an \eqdl by applying it $k$
times on the same witness. It is often more efficient to combine the application
in one as in \cite{Chaum92,Chow10}, or more generally as composition of $\Sigma$-protocols, here with two logarithms and two generators $g_1$ and $g_2$. 
Peggy  wants to prove that she knows $x$ such that $v = g_1^x$ and $w = g_2^x$:
\begin{enumerate}
\item Peggy chooses $r$ at random and sends $\lambda = g_1^r$ and $\mu = g_2^r$
  to Victor.
\item Victor chooses a challenge $c$ at random and sends it to Peggy.
\item Peggy computes $s = (r + c \cdot x) \mod q$ and sends it to Victor.
\item Victor tests if $g_1^s = \lambda \cdot v^c$ and $g_2^s = \mu \cdot w^c$.
\end{enumerate}
This protocol remains malleable, and the previous attacks are
still valid since the response remains of the form $r+c\cdot x$.

\subsubsection{Countermeasures}\label{ssec:counter}

Direct countermeasures to the above attacks are to use non-interactive and/or non-malleable proofs:
\begin{itemize}
\item An interactive protocol can be converted into a non-interactive one using
  the Fiat-Shamir heuristic~\cite{Fiat86}. 
\item Also the first \pdl by \cite{Chaum86} uses bit-flipping, and more
  generally non-malleable protocols like~\cite{Fischlin2009} could be used.
\end{itemize}

We will show in the following that if the proofs proposed in the original paper are not converted into non-interactive proofs, there is an attack on privacy. Note that even if non-interactive non-malleable zero-knowledge proofs are used, a malicious attacker in control of the network can nonetheless recover any bidder's bid as the messages are not authenticated, as we show in Section~\ref{sec:attackauth}.

\section{Attacking the fully private computations}\label{sec:attackcomp}
The first attack we present uses some algebraic
properties of the computations performed during the protocol
execution.

\subsection{Analysis of the outcome computation}\label{ssec:outcome}

The idea is to analyze the computations done in
Step~\ref{it:outcome} of the protocol. Consider the following example
with three bidders and three possible prices.  Then the first bidder computes
\[
\begin{array}{rrrrclcll}
 \gamma_{1 1}^1 = ( &(\alpha_{12} \cdot  \alpha_{13}  \cdot & \alpha_{22}
 \cdot \alpha_{23} \cdot & \alpha_{32} \cdot \alpha_{33}) & \cdot  & (1) & \cdot  & (1) & )^{m_{11}^1} \\
 \gamma_{1 2}^1 = ( & (\alpha_{13}  \cdot & \alpha_{23}\cdot & \alpha_{33}) & \cdot  & (\alpha_{11}) & \cdot  & (1) & )^{m_{12}^1} \\
 \gamma_{1 3}^1 = ( & & & (1) & \cdot  & (\alpha_{11} \cdot \alpha_{12}) & \cdot  & (1) & )^{m_{13}^1} \\
 \gamma_{2 1}^1 = ( & (\alpha_{12} \cdot \alpha_{13} \cdot &
 \alpha_{22}  \cdot \alpha_{23}  \cdot& \alpha_{32}  \cdot \alpha_{33}) & \cdot  & (1) & \cdot  & (\alpha_{11}) & )^{m_{21}^1} \\
 \gamma_{2 2}^1 = ( & (\alpha_{13} \cdot & \alpha_{23} \cdot & \alpha_{33}) & \cdot  & (\alpha_{21}) & \cdot  & (\alpha_{12}) & )^{m_{22}^1} \\
 \gamma_{2 3}^1 = ( & & & (1) & \cdot  & (\alpha_{21} \cdot \alpha_{22}) & \cdot  & (\alpha_{13}) & )^{m_{23}^1} \\
 \gamma_{3 1}^1 = ( & (\alpha_{12} \cdot \alpha_{13} \cdot & \alpha_{22} \cdot \alpha_{23} \cdot & \alpha_{32} \cdot \alpha_{33}) & \cdot  & (1) & \cdot  & (\alpha_{11} \cdot \alpha_{21}) & )^{m_{31}^1} \\
 \gamma_{3 2}^1 = ( & (\alpha_{13} \cdot & \alpha_{23} \cdot &  \alpha_{33}) & \cdot  & (\alpha_{31}) & \cdot  & (\alpha_{12} \cdot \alpha_{22}) & )^{m_{32}^1} \\
 \gamma_{3 3}^1 = ( & & & (1) & \cdot  & (\alpha_{31} \cdot \alpha_{32}) & \cdot  & (\alpha_{13} \cdot \alpha_{23}) & )^{m_{33}^1} \\
\end{array}
\]
The second and third bidder do the same computations, but using different random values $m_{ij}^a$. Since each $\alpha_{ij}$ is either the encryption of $1$ or $Y$, for example the value $\gamma_{2 2}^1$ will be an encryption of $1$ only if
\begin{itemize}
 \item nobody submitted a higher bid (the first block) and 
 \item bidder 2 did not bid a lower bid (the second block) and 
 \item no bidder with a lower index submitted the same bid (the third block).
\end{itemize}
If we ignore the exponentiation by $m_{i j}^a$, each $\gamma_{i j}^a$
is the encryption of the product of several $b_{i j}$'s. Each $b_{i
  j}$ can be either 1 or $Y$, hence $({\gamma_{i j}^a})^{-m_{i j}^a}$
will be the encryption of a value $Y^{l_{i j}}$, where $0 \leq l_{i j}
\leq n$. The lower bound of $l_{i j}$ is trivial, the upper bound
follows from the observation that each $\alpha_{i j}$ will be used at
most once, and that each bidder will encrypt $Y$ at most once.

Assume for now that we know all $l_{i j}$. We show next that this is
sufficient to obtain all bids.
Consider the function $f$ which takes as input the following vector\footnote{By abuse of notation we write $log_s \left( \begin{matrix} x_1, & \ldots, & x_n \\ \end{matrix} \right)$ for $\left( \begin{matrix} log_s(x_1), & \ldots, & log_s(x_n) \\ \end{matrix} \right)$.}:
$
 b = \log_ Y\left( \left( \begin{array}{ccccccccccccc}
  b_{11}, &
  \ldots, &
  b_{1k}, &
  \hspace{2mm} &
  b_{21}, &
  \ldots, &
  b_{2k}, &
  \hspace{2mm} &
  \ldots, &
  \hspace{2mm} &
  b_{n1}, &
  \ldots, &
  b_{nk} \\
 \end{array} \right)^T \right)
$, 
and returns the values $l_{i j}$. The input vector is thus a vector of all bid-vectors, where $1$ is replaced by $0$ and $Y$ by $1$.
Consider our above example with three bidders and three possible prices, then we have:
\\\[
 b = \log_Y \left( \left( \begin{array}{ccccccccccc}
  b_{11}, &
  b_{12}, &
  b_{13}, &
  \hspace{2mm} &
  b_{21}, &
  b_{22}, &
  b_{23}, &
  \hspace{2mm} &
  b_{31}, &
  b_{32}, &
  b_{33} \\
 \end{array} \right)^T \right).
\]
A particular instance where bidder 1 and 3 submit price 1, and bidder 2 submits price 2 would then look as:
$
 b = \left( \begin{array}{ccccccccccc}
  1, &
  0, &
  0, &
  \hspace{2mm} &
  0, &
  1, &
  0, &
  \hspace{2mm} &
  1, &
  0, &
  0 \\
 \end{array} \right)^T
$. 
Hence only the factors $\alpha_{11}$, $\alpha_{22}$ and $\alpha_{31}$ are encryptions of $Y$, all other $\alpha$'s are encryptions of $1$. By simply counting how often the factors $\alpha_{11}$, $\alpha_{22}$ and $\alpha_{31}$ show up in each equation as described above, we can compute the following result:
$f(b)=
 \left( \begin{array}{ccccccccccc}
  1, &
  1, &
  1, &
  \hspace{2mm} &
  2, &
  0, &
  1, &
  \hspace{2mm} &
  2, &
  1, &
  1 \\
 \end{array} \right)^T$.
Note that since we chose the input of $f$ to be a bit-vector, we have to simply count the ones (which correspond to $Y$'s) in particular positions in $b$, where the positions are determined by the factors inside $\gamma_{i j}^a$. Hence we can express $f$ as a matrix, i.e. $f(b) = M \cdot b$ for the following matrix $M$:
\begin{equation*}
 f(b) = M \cdot b = \left[ \begin{array}{ccccccccccc}
  0 & 1 & 1 & \hspace{2mm} & 0 & 1 & 1 & \hspace{2mm} & 0 & 1 & 1 \\
  1 & 0 & 1 & & 0 & 0 & 1 & & 0 & 0 & 1 \\
\vspace{2mm}
  1 & 1 & 0 & & 0 & 0 & 0 & & 0 & 0 & 0 \\
  1 & 1 & 1 & & 0 & 1 & 1 & & 0 & 1 & 1 \\
  \mathbf{0} & \mathbf{1} & \mathbf{1} & & \mathbf{1} & \mathbf{0} & \mathbf{1} & & \mathbf{0} & \mathbf{0} & \mathbf{1} \\
\vspace{2mm}
  0 & 0 & 1 & & 1 & 1 & 0 & & 0 & 0 & 0 \\
  1 & 1 & 1 & & 1 & 1 & 1 & & 0 & 1 & 1 \\
  0 & 1 & 1 & & 0 & 1 & 1 & & 1 & 0 & 1 \\
  0 & 0 & 1 & & 0 & 0 & 1 & & 1 & 1 & 0 \\
 \end{array} \right] \cdot  \left( \begin{matrix}
  1 \\
  0 \\
\vspace{2mm}
  0 \\
  0 \\
  \mathbf{1} \\
\vspace{2mm}
  0 \\
  1 \\
  0 \\
  0 \\
 \end{matrix} \right) = \left( \begin{matrix}
  1 \\
  1 \\
\vspace{2mm}
  1 \\
  2 \\
  \mathbf{0} \\
\vspace{2mm}
  1 \\
  2 \\
  2 \\
  1 \\
 \end{matrix} \right)
\end{equation*}
To see how the matrix $M$ is constructed, consider for example 
$({\gamma_{2 2}^a})^{-m_{2 2}^a} = (\alpha_{13}  \cdot \alpha_{23}  \cdot \alpha_{33}) \cdot (\alpha_{21}) \cdot (\alpha_{12})$
which corresponds to the \textbf{second row} in the second vertical block:
\begin{itemize}
 \item $\alpha_{12}$ and $\alpha_{13}$; hence the two ones at position 2 and 3 in the first horizontal block
 \item $\alpha_{21}$ and $\alpha_{23}$; hence the two ones at position 1 and 3 in the second horizontal block
 \item $\alpha_{33}$; hence the one at position 3 in the third horizontal block
\end{itemize}
More generally, we can see that each $3 \times 3$ block consists of potentially three parts:
\begin{itemize}
 \item An upper triangular matrix representing all bigger bids.
 \item On the diagonal we add a lower triangular matrix representing a lower bid by the same bidder,
 \item In the lower left half we add an identity matrix representing a bid at the current price by a bidder with a lower index.
\end{itemize}
This corresponds exactly to the structure of the products inside each $\gamma_{i j}^a$. It is also equivalent to formula (1) in Section 4.1.1 of the original paper~\cite{Brandt06} without the random vector $R_k^*$.
In the following we prove that the function $f$ is injective. We
then discuss how this function can be efficiently inverted (i.e. how
to compute the bids when knowing all $l_{i j}$'s).

\subsection{Linear algebra toolbox}

Let $I_k$ be the $k\times k$ identity matrix;

let $L_k$ be a lower $k\times k$ triangular matrix  with
zeroes on the diagonal, ones in the lower part and zeroes elsewhere;
and let $U_k$ be an upper $k\times k$ triangular matrix  with zeroes on the
diagonal, ones in the upper part, and zeroes elsewhere:

\[
 \ I_k = 
\begin{bmatrix}
1 & 0 & \cdots & 0 \\
0 & \ddots & \ddots & \vdots \\
\vdots & \ddots & \ddots & 0 \\
0 & \cdots & 0 & 1 \end{bmatrix}
~~~~
    L_k= 
\begin{bmatrix}
0 & 0 & \cdots & 0 \\
1 & \ddots & \ddots & \vdots \\
\vdots & \ddots & \ddots & 0 \\
1 & \cdots & 1 & 0 \end{bmatrix}
~~~~
  U_k = 
\begin{bmatrix}
0 & 1 & \cdots & 1 \\
0 & \ddots & \ddots & \vdots \\
\vdots & \ddots & \ddots & 1 \\
0 & \cdots & 0 & 0 \end{bmatrix}
\]

By
abuse of notation we use $I$, $L$ and $U$ to denote respectively $I_k$,
$L_k$ and $U_k$.
For a $k\times k$-matrix $M_k$ we define $(M_k)^r= M \cdots M$ ($r$ times) and $(M_k)^0= I_k$. Let $(e_1, \ldots, e_k) $ be the canonical basis.

\begin{lemma}\label{lem:nilpotent}
Matrices $L_k$ and $U_k$ are nilpotent, {\it i.e.}  $(U_k)^k =0$ and
$(L_k)^k =0$.
\end{lemma}

\begin{lemma}\label{lem:somme}
If $\sum_{j=1}^{k} w_j = 1 $ then we have $L_k \cdot w =  (1,\ldots,1)^T - (I_k +
U_k)\cdot w$.
\end{lemma}

\begin{proof}
First note that since  $\sum_{j=1}^{k} w_j = 1$,
$$L_k \cdot w=
\begin{bmatrix}
0 & 0 & \cdots & 0 \\
1 & \ddots & \ddots & \vdots \\
\vdots & \ddots & \ddots & 0 \\
1 & \cdots & 1 & 0 \end{bmatrix}
  \cdot  
 \begin{bmatrix}
w_1 \\
\vdots \\
w_k
\end{bmatrix} =  \begin{bmatrix}
0\\
w_1\\
w_1 + w_2\\
\vdots \\
\sum_{j=1}^{k-1}w_j
\end{bmatrix}
= 
 \begin{bmatrix}
1 - \sum_{j=1}^{k} w_j \\
1 - \sum_{j=2}^{k}w_j\\
\vdots \\
1- w_k
\end{bmatrix}
 $$
\noindent On the other hand, if we let $ \mathbf{1}=(1,\ldots,1)^T$, we have also:
$$  \mathbf{1}- (I_k +
U_k)\cdot w =  \mathbf{1}- 
\begin{bmatrix}
1 & 1 & \cdots & 1 \\
0 & 1 & \ddots & \vdots \\
\vdots & \ddots & \ddots & 1 \\
0 & \cdots & 0 & 1 \end{bmatrix}
\cdot
 \begin{bmatrix}
w_1 \\
\vdots \\
w_k
\end{bmatrix}= \begin{bmatrix}
1 - \sum_{j=1}^{k} w_j \\
1 - \sum_{j=2}^{k}w_j\\
\vdots \\
1- w_k
\end{bmatrix}
$$  

\end{proof}

\begin{lemma}\label{lem:luz}
For $z = e_i - e_j$, we have that $(L_k+U_k)\cdot z = -z$.
\end{lemma}
\begin{proof}
If $i=j$, then $z=0$ and the results is true.
Suppose w.l.o.g. that $i>j$ (otherwise we just prove the result
for $-z$). Then 
$U_k\cdot (e_i-e_j) = \sum_{s=1}^{i-1} e_s - \sum_{s=1}^{j-1} e_s =
\sum_{s=j}^{i-1} e_s.$
Similarly
$L_k\cdot (e_i-e_j) = \sum_{s=i+1}^{k} e_s - \sum_{s=j+1}^{k} e_s = \sum_{s=j+1}^{i} -e_s.$
Therefore
$(L_k+U_k)\cdot(e_i-e_j)=\sum_{s=j}^{i-1} e_s -  \sum_{s=j+1}^{i} e_s = e_j -e_i = -z.$
\end{proof}

\subsection{How to recover the bids when knowing the $l_{i j}$'s}

As discussed above, we can represent the function $f$ as a matrix
multiplication. Let $M$ be the following square matrix of size $nk
\times nk$:
\begin{equation*}
M=\left[
\begin{matrix}
(U+L) & U   &  \ldots & \ldots & U \\
(U+I) & (U+L) & U & \ldots & U \\
\vdots & \ddots & \ddots & \ddots & \vdots \\
(U+I) & \ldots & (U+I) & (U+L) & U \\
(U+I) & \ldots & \ldots & (U+I) & (U+L) \\
\end{matrix}
\right]\text{.~Then}~f (b) = M \cdot b\text{.}
\end{equation*}

The function takes as input a vector composed of $n$ vectors, each of
$k$ bits. It returns the $nk$ values $l_{i j}$, $1 \leq i \leq n$ and $1 \leq j \leq k$. As explained above, the structure of the matrix is defined by the formula that computes $\gamma_{i j}^a$, which consists essentially of three factors: first we multiply all $\alpha_{i j}$ which encode bigger bids (represented by the matrix $U$), then we multiply all $\alpha_{i j}$ which encode smaller bids by the same bidder (represented by adding the matrix $L$ on the diagonal), and finally we multiply by all $\alpha_{i j}$ which encode the same bid by bidders with a smaller index (represented by adding the matrix $I$ on the lower triangle of $M$). In our encoding there will be a ``$1$'' in the vector for each $Y$ in the protocol, hence $f$ will count how many $Y$s are multiplied when computing $\gamma_{i j}^a$. Using this representation we can prove the following theorem. 
\begin{theorem}
$f$ is injective on valid bid vectors, {\it i.e.} for two different correct bid vectors 
$u=[u_1,\ldots, u_k]^T$
and
$v = [v_1,\ldots,v_k]^T
$
with $u \neq v$ we have $M \cdot u \neq M \cdot v$.
\end{theorem}
\begin{proof}
Let $u$ and $v$ be two correct bid vectors such that $u \neq v$. We
want to prove that $M \cdot u \neq M \cdot v$. We make a proof by contradiction,
hence we
assume that $M \cdot u =M \cdot v$ or that $M \cdot (u-v)=0$.  Because $u$
and $v$ are two correct bid vectors, each one of them is an element of the
canonical basis $(e_1, \ldots, e_k)$, i.e. $u=e_i$ and $v=e_j$, as
shown in Section~\ref{ssec:outcome}.
We denote $u-v$ by $z$, and consequently $z=e_i-e_j$.
Knowing that $M \cdot z=0$, we prove by induction on $a$ that for all $a$ the
following property $P(a)$ holds:

\[
P(a): \forall l,  0 < l \leq a, diag(U^{k-l})\cdot z = 0 
\]
where $diag(U^{k-x})$ is a $nk \times nk$ block diagonal matrix containing only
diagonal blocks of the same matrix $U^{k-x}$. 
The validity of $P(k)$ proves in particular that $diag(U^{0})\cdot z_l=0$, i.e. $z=0$ which contradicts our hypothesis.
\begin{itemize}
\item Case $a = 1$: we also prove this base case by induction,
  i.e. for all $b \geq 1$ the property $Q(b)$ holds, where:
 $$Q(b): \forall m,  0< m \leq  b, U^{k-1}\cdot z_m = 0$$
 which gives us that $U^{k-1}\cdot z = 0$.
  \begin{itemize}
   \item Base case $b = 1$: 
    We start by looking at the multiplication of the first row of $M$ with $z$. We obtain:
    $(L+U)\cdot z_1 + U \cdot (z_2+ \ldots + z_k) = 0.$
    We can multiply each side by $U^{k-1}$, and use
    Lemma~\ref{lem:luz} to obtain:
    $U^{k-1}\cdot [ -z_1 + U^k\cdot (z_2+ \ldots + z_k)] = 0 .$
    Since $U$ is nilpotent, according to Lemma~\ref{lem:nilpotent}
    the latter gives $-U^{k-1}\cdot z_1 = 0$.
    Hence we know $Q(1): U^{k-1}\cdot z_1=0$, i.e. the last entry of $z_1$ is 0.
   \item Inductive step $b + 1$: assume $Q(b)$. Consider now the multiplication of the $(b+1)$-th row of the matrix $M$:
    \\$(U+I)\cdot z_1 + \ldots + (U+I)\cdot z_b + (L+U)\cdot z_{b+1} +
    U\cdot  (z_{b+2} + \ldots + z_k) = 0.$
    Then by multiplying by $U^{k-1}$ and using Lemma~\ref{lem:luz}
    we obtain:
    \\$U^{k-1}\cdot [(U+I) \cdot z_1 + \ldots + (U+I)\cdot z_b -z_{b+1} + U \cdot(z_{b+2} + \ldots + z_k)] = 0.$
    Since $U$ is nilpotent according to Lemma~\ref{lem:nilpotent} we have 
    $U^{k-1}\cdot z_1 + \ldots + U^{k-1} \cdot z_b - U^{k-1}\cdot z_{b+1} = 0.$
    Using the fact that for all $m < b$ we have $U^{k-1}\cdot z_m=0$, the
    latter gives $- U^{k-1}\cdot  z_{b+1} = 0$.
   \end{itemize}
 \item Inductive step $a + 1$: assume $P(a)$. By induction on $b \geq
   1$ we will show that $Q'(b)$ holds, where 
$$Q'(b): \forall m, 0 < m \leq b, U^{k-(a+1)}\cdot z_m = 0$$ which gives us that
$U^{k-(a+1)}\cdot z = 0$, i.e. $P(a+1)$.
  \begin{itemize}
   \item Base case $b = 1$: 
    Consider the multiplication of the first row with $U^{k-(a+1)}$:
    $U^{k-(a+1)}\cdot [ (L+U) \cdot z_1 + U\cdot (z_2+ \ldots + z_k)] = 0$
    which can be rewritten as
    $-U^{k-(a+1)} \cdot  z_1 + U^{k-a} \cdot  ( z_2 + \ldots + z_k)] = 0 .$ 
    Using $U^{k-a} \cdot z_l = 0$ for all $l$, we can conclude that
    $-U^{k-(a+1)}\cdot  z_1 = 0 ,$ i.e. $Q'(1)$ holds.
   \item Inductive step $b + 1$: assume $Q'(b)$. Consider now the $(b+1)$-th row of the matrix $M$:
    \\$(U+I)\cdot z_1 + \ldots + (U+I)\cdot z_b + (L+U)\cdot z_{b+1} + U\cdot  (z_{b+2} + \ldots + z_k) = 0.$
    Then by multiplying by $U^{k-(a+1)}$ and using
    Lemma~\ref{lem:luz} we obtain:
    \\$U^{k-(a+1)}\cdot [(U+I)\cdot z_1 + \ldots + (U+I)\cdot z_b + -z_{b+1} + U
    \cdot  (z_{b+2} + \ldots + z_k)] = 0.$
    Using $U^{k-a}\cdot  z_l = 0$ for all $l$, we can conclude that 
    $U^{k-(a+1)} \cdot  z_1 + \ldots + U^{k-(a+1)}\cdot   z_b - U^{k-(a+1)} \cdot  z_{b+1} = 0.$
    Now, for all $m < b$, we have $U^{k-(a+1)} \cdot z_m=0$, so that $-U^{k-(a+1)} \cdot  z_{b+1}=0;$ i.e. $Q'(b+1)$ holds.
\qed  \end{itemize}
\end{itemize} 
\end{proof}
This theorem shows that if there is a constellation of bids that led to certain values $l_{i j}$, this constellation is unique. Hence we are able to invert $f$ on valid outputs. We will now show that this can be efficiently done.

\subsubsection{An efficient algorithm}\label{ssec:resol}

Our aim is to solve the following linear system: $M\cdot x = l$.
We will use the same steps we used for the proof of injectivity to
solve this system efficiently.

Consider
the $r$-th block of size $k$ of the latter system. 
We have $x_r = (x_{r,1}, x_{r,2}, \ldots,$ $x_{r,k})$

and the $r$-th block of $M\cdot x$~is
\begin{multline*}
(U+I) x_1+\ldots+(U+I)x_{r-1}+(L+U)x_r+Ux_{r+1}+\ldots +Ux_n \\= 
\textstyle U(\sum_{i=1}^{n}x_i)+(\sum_{i=1}^{r-1}x_i)+L x_r.
\end{multline*}
As the $r$-th block of $l$ is $l_r$, we thus have: 
$U\left(\sum_{i=1}^{n}x_i\right)+\sum_{i=1}^{r-1}x_i +L x_r=l_r$.

Using Lemma~\ref{lem:somme}, with $w_j=x_{r,j}$ for $j=1..k$, we can exchange $L$ in the latter to get
$
U\left(\sum_{i=1}^{n}x_i\right)+\sum_{i=1}^{r-1}x_i
+\left( \mathbf{1}- \left(I +U\right)x_r\right)=l_r.$
Hence, $
U\left(\sum_{i=1}^{n}x_i\right)+\sum_{i=1}^{r-1}x_i
+\mathbf{1}- x_r-Ux_r=l_r$, 
so that we now have:\\
\begin{equation}\label{eq:algo}
\begin{cases}
x_1 = \mathbf{1}-l_1+U\left(\sum_{i=2}^{n}x_i\right) \\
x_r = \mathbf{1}-l_r+\sum_{i=1}^{r-1}x_i+U\left(\sum_{i=1,i\neq
    r}^{n}x_i\right) & \text{~if~} 1<r\leq n
\end{cases}
\end{equation}

This gives us a formula to compute the values of $x_{i,j}$, starting
with the last element of the first block $x_{1,k}$. Then we can
compute the last elements of all other blocks $x_{2,k}, \ldots,
x_{n,k}$, and then the second to last elements $x_{1,k-1}, \ldots,
x_{n,k-1}$, etc.

The idea is to project the above Equation~(\ref{eq:algo}) on
the $t$-th coordinate. Then, the $t$-th row of $U$ has ones only
starting at index $t+1$, and thus the $t$-th row of $U z$ involves
only the elements $z_{t+1},\ldots,z_k$.
We thus have:
$e_t^T
U\left(\sum_{i=1,i\neq r}^{n}x_i\right) = \sum_{j=t+1}^{k}\sum_{i=1,i\neq r}^{n}x_{i,j}
$ for $t<k$ and $e_k^TU=0$.
Now $e_t^T x_r = x_{r,t}$, $e_t^T l_r = l_{r,t}$ and $e_t^T
\mathbf{1}=1$. Hence, we therefore get the
following where at row $t$, the right hand side involves only already
computed values: 
\begin{equation}\label{eq:algoform}
\begin{dcases}
x_{1,k} = 1-l_{1,k}\\
x_{r,k} = 1-l_{r,k}+\sum_{i=1}^{r-1}x_{i,k} & \text{~if~} 1<r\leq n \\
x_{1,t} = 1-l_{1,t}+\sum_{j=t+1}^{k}\sum_{i=2}^{n}x_{i,j}& \text{~if~}
1\leq t<k \\
x_{r,t} = 1-l_{r,t}+\sum_{i=1}^{r-1}x_{i,t}+
\sum_{j=t+1}^{k}\sum_{i=1,i\neq r}^{n}x_{i,j}& \text{~if~} 1\leq t<k
\text{~and~} 1<r\leq n
\end{dcases}\end{equation}

\subsubsection{Complexity Analysis.}

To obtain all values, we have to apply the above
Formula~(\ref{eq:algoform}) for each $1\leq r \leq n$ and $1\leq t \leq k$, hence
we can bound the arithmetic cost by:
$$\sum_{r=1}^{n} \sum_{t=1}^{k} \left(r+(k-t)n\right)= 
\frac{1}{2}n^2k^2+o\left(n^2k^2\right)$$
This is efficient enough to be computed on a standard PC for realistic
values of $n$ (the number of bidders) and $k$ (the number of possible
bids). Those could be less than a hundred bidders with a thousand
different prices, thus requiring about the order of only some
giga arithmetic operations.  
It is anyway the order of magnitude of the number of operations
required to compute all the encrypted bids.

\subsection{Attack on the random noise: how to obtain the $l_{i j}$'s}\label{ssec:noiserem}

In the previous section we showed that knowing the $l_{i j}$'s allows
us the efficiently break the privacy of all bidders. Here is how
to obtain the $l_{i j}$'s.

The seller will learn all 
$v_{i j} = \left( Y^{l_{i j}} \right)^{(\sum_{h=1}^n m_{i j}^h)}$
at the end of the protocol. Since the $m_{i j}^h$ are randomly chosen,
this will be a random value if $l_{i j} \neq 0$. However a malicious
bidder (``Mallory'', of index $a$) can cancel out the $m_{i j}^h$ as follows: in Step~\ref{it:outcome} of the protocol each bidder will compute his $\gamma_{i j}^{a}$ and $\delta_{i j}^{a}$. Mallory waits until all other bidders have published their values (the protocol does not impose any synchronization or special ordering) and then computes his values $\gamma_{i j}^{\omega}$ and $\delta_{i j}^{\omega}$ as:
\begin{center}$ \gamma_{i j}^{\omega} = \left( \left( \prod_{h=1}^n \prod_{d=j+1}^k \alpha_{h d} \right) \cdot \left( \prod_{d=1}^{j-1} \alpha_{i d} \right) \cdot \left( \prod_{h=1}^{i-1} \alpha_{h j} \right) \right) \cdot \left( \prod_{k \neq \omega} \gamma_{i j}^{k} \right)^{-1}
$
\\$
 \delta_{i j}^{\omega} = \left( \left( \prod_{h=1}^n \prod_{d=j+1}^k \beta_{h d} \right) \cdot \left( \prod_{d=1}^{j-1} \beta_{i d} \right) \cdot \left( \prod_{h=1}^{i-1} \beta_{h j} \right) \right) \cdot \left( \prod_{k \neq \omega} \delta_{i j}^{k} \right)^{-1}
$\end{center}
The first part is a correct encryption of $Y^{l_{i j}}$, with
$m_{ij}^{\omega}=1$ for all $i$ and $j$. The second
part is the inverse of the product of all the other bidders $\gamma_{i
  j}^{k}$ and $\delta_{i j}^{k}$, and thus it will eliminate the
random exponents. Hence after decryption the seller obtains $v_{i j} =
Y^{l_{i j}}$, where $l_{i j}<n$ for a small $n$. He can compute $l_{i
  j}$ by simply (pre-)computing all possible values $Y^r$ and testing
for equality. This allows the seller to obtain the necessary values
and then to use the resolution algorithm to obtain each bidder's
bid. Note that although we changed the intermediate values, the output
still gives the correct result (i.e. winning bid).
Therefore, the attack might even be unnoticed by the other
participants. 
Note also that choosing a different $Y_i$ per bidder does not prevent
the attack, since all the $Y_i$ need to be public in order to prove
the correctness of the bid in Step~\ref{it:bid} of the protocol.

However the protocol requires Mallory to prove that $\gamma_{i j}^{\omega}$
and $\delta_{i j}^{\omega}$ have the same exponent. This is obviously the
case, but Mallory does not know the exact value of this exponent.
Thus it is impossible for him to execute the proposed
zero-knowledge protocol directly.

In the original paper~\cite{Brandt06} the malleable interactive
proof of~\cite{Chaum92}, presented in Section~\ref{ssec:PoK}, is used
to prove the correctness of $\gamma_{i j}^a$ and $\delta_{i j}^a$ in
Step~\ref{it:outcome} of the protocol.

If this proof is not converted into a non-interactive proof, then Mallory is able to fake it as follows.

\subsection{Proof of equality of the presented outcomes}\label{ssec:eqdlattack}
Note that we can rewrite $\gamma_{i j}^{\omega}$ and $\delta_{i j}^{\omega}$ as:
\[\textstyle
 v = \gamma_{i j}^{\omega} = {\underbrace{ \left( \left( \prod_{h=1}^n \prod_{d=j+1}^k \alpha_{h d} \right) \cdot \left( \prod_{d=1}^{j-1} \alpha_{i d} \right) \cdot \left( \prod_{h=1}^{i-1} \alpha_{h j} \right) \right)}_{g_1}}^{1-\left(\sum_{k \neq \omega} m_{i j}^{k}\right)}
\]
\[\textstyle
 w = \delta_{i j}^{\omega} = {\underbrace{ \left( \left( \prod_{h=1}^n \prod_{d=j+1}^k \beta_{h d} \right) \cdot \left( \prod_{d=1}^{j-1} \beta_{i d} \right) \cdot \left( \prod_{h=1}^{i-1} \beta_{h j} \right) \right)}_{g_2}}^{1-\left(\sum_{k \neq \omega} m_{i j}^{k}\right)} 
\]
When Mallory, the bidder $m$, is asked by Victor for a proof of correctness of his
values, he starts by asking all other bidders for proofs to initialize
the man-in-the-middle attack of Figure~\ref{fig:exm}. Each of them
answers with values $\lambda_o = g_1^{z_o}$ and $\mu_o =
g_2^{z_o}$. Mallory can then answer Victor with values
$\lambda=\prod_{o} \lambda_o^{-1}$ and
$\mu=\prod_{o} \mu_o^{-1}$, where  $o \in  ([1,n] \setminus
m)$. Victor then sends a challenge $c$, which Mallory simply forwards
to the other bidders. They  answer with $r_o = z_o + c \cdot m_{i
  j}^o$, and Mallory sends $r = c - \sum_o r_o$ to Victor, who can
check that $g_1^r = \lambda \cdot v^c$ and $g_2^r = \mu \cdot w^c$. If
the other bidders did their proofs correctly, then Mallory's proof
will appear valid to Victor: 
\[\textstyle
\lambda \cdot v^c = \prod_o \lambda_o^{-1} \cdot {\left( g_1^{1 - \left(\sum_{o} m_{i j}^{o}\right)} \right)}^c = \prod_o g_1^{-z_o} \cdot g_1^{c-c\left(\sum_{o} m_{i j}^{o}\right)} = g_1^{c-\sum_{o} \left(z_o + c m_{i j}^{o}\right)}
\]
\[\textstyle
\mu \cdot w^c = \prod_o \mu_o^{-1} \cdot {\left( g_2^{1 - \left(\sum_{o} m_{i j}^{o}\right)} \right)}^c = \prod_o g_2^{-z_o} \cdot g_2^{c-c\left(\sum_{o} m_{i j}^{o}\right)} = g_2^{c-\sum_{o} \left(z_o + c m_{i j}^{o}\right)}
\]
Hence in the case of malleable interactive zero-knowledge proofs Mallory is able to modify the values $\gamma_{i j}^{\omega}$ and $\delta_{i j}^{\omega}$ as necessary, and even prove the correctness using the bidders. Hence the modifications may stay undetected and the seller will be able to break privacy.

\subsection{The complete attack and countermeasures}\label{ssec:countermitm}
Putting everything together, the attack works as follows:
\begin{enumerate}
 \item The bidders set up the keys as described in the protocol.
 \item They encrypt and publish their bids.
 \item They compute $\gamma_{i j}^{h}$ and $\delta_{i j}^{h}$ and publish them.
 \item Mallory, who is a bidder himself, waits until all other bidders have published their values. He then computes his values as defined above, and publishes them.
 \item If he is asked for a proof, he can proceed as explained above
   in Section~\ref{ssec:eqdlattack}.
 \item The bidders (including Mallory) jointly decrypt the values.
 \item The seller obtains all $Y^{l_{i j}}$'s. He can then compute
   the $l_{i j}$'s by testing at most $n$ possibilities.
 \item Once he has all values, he can invert the function $f$ as explained above.
 \item He obtains all bidders bids.
\end{enumerate}
Again, note that for all honest bidders, this execution will look normal, so they might not even notice that an attack took place.

To prevent this attack, one could perform the following actions:
\begin{itemize}
\item To counteract the removal of the noise of Section~\ref{ssec:noiserem},
  the bidders could check whether the product of the $\gamma_{i,j}^a$
  for all bidders $a$ is equal to the product of the $\alpha_{hd}$
  without any noise (exponent is $1$). Unfortunately, the
  man-in-the-middle attack generalizes to any exponent as shown in
  Figure~\ref{fig:generic}. Therefore the attacker could use a
  randomly chosen exponent only known to him.
\item As mentioned above, another countermeasure is the use of
  non-interactive, non-malleable proofs of knowledge. 
  In this case, we will show in Section~\ref{sec:attackauth} that it is still possible to attack
  a targeted bidder's privacy.
\end{itemize}

\section{Attacking verifiability}\label{sec:attackverif}
Brandt claims that the protocol is verifiable as the parties have to provide zero-knowledge proofs for their computations, however there are two problems.

\subsection{Exceptional values}
First, a winning bidder cannot verify if he actually won. 
To achieve privacy, the protocol hides all outputs of $v_{aj}$ except
for the entry containing ``1''\footnote{Note that the protocol contains a mechanism to resolve ties, i.e. there should always be exactly one entry equal to 1, even in the presence of ties.}. This is done by exponentiation with
random values $m_{ij}^a$ inside all entries $\gamma^a_{ij}$ and
$\delta^a_{ij}$, i.e. by computing $x_{ij}^{\sum_a m_{ij}^a}$ where
$x_{ij}$ is the product of some $\alpha_{ij}$ as specified in the
protocol. 
If $x_{ij}$ is one, $x_{ij}^m$ will still return one for any $m$, and
in principle something different from one for
any other value of $x_{ij}$. Now, the random values $m^a_{ij}$ may add up
to zero (mod $q$), hence the returned value will be
$x_{ij}^m=x_{ij}^0=1$ and the bidder will conclude that he won,
although he actually lost ($x_{ij} \neq 1$). 
Hence simply verifying the proofs is not sufficient to be convinced
that the observed outcome is correct. 
For the same reason the seller might observe two or more ``1''-values,
even though all proofs are correct. In such a situation he is unable to decide which
bidder actually won since he cannot determine which ``1''s correspond to a real bids, and hence which bid is the highest real bid. If two ``1''s correspond to real bids, he could even exploit such a situation to his
advantage: he can tell both bidders that they won and take
money from both, although there is only one good to sell -- this is normally prohibited by the protocol's tie-breaking mechanism. 
If the bidders do not exchange additional data there is no way for them to discover that something went wrong, since the seller is the only party having access to all values. 

A solution to this problem could work as follows: when computing the
$\gamma^a_{ij}$ and $\delta^a_{ij}$, the bidders can check if the
product $$x_{ij} = \left( \prod_{h=1}^n \prod_{d=j+1}^k \alpha_{h d} \right) \cdot \left( \prod_{d=1}^{j-1} \alpha_{i d} \right) \cdot \left( \prod_{h=1}^{i-1} \alpha_{h j} \right)$$ is equal to one -- if yes, they restart the protocol using different keys and random values. If not, they continue, and check if $\prod_a \gamma^a_{ij} = 1$. If yes, they choose different random values $m_{ij}^a$ and re-compute the $\gamma^a_{ij}$ and $\delta^a_{ij}$, otherwise they continue. Since the probability of the random values adding up to zero is low, this will rapidly lead to correct values.

\subsection{Different private keys}
Second, the paper does not precisely specify the proofs that have to
be provided in the joint decryption phase. 
If the bidders only prove that they use the same private key on all
decryptions \emph{and not also that it is the one they used to
generate their public key}, they may use a wrong one. 
This will lead to a wrong decryption where with very high probability
no value is ``1'', as they will be random. 
Hence all bidders will think that they lost, thus allowing a malicious
bidder to block the whole auction, as no winner is determined.
Hence, if we assume that the verification test consists in verifying
the proofs, a bidder trying to verify that he lost using the proofs
might perform the verification successfully, although the result is
incorrect and he actually won -- since he would have observed a ``1''
if the vector had been correctly decrypted. 

This problem can be addressed by requiring the bidders to also prove
that they used the same private key as in the key generation phase.

\section{Attacks using the lack of
  authentication}\label{sec:attackauth}
The protocol as described in the original paper does not include any authentication of the messages. This means that an attacker in control of the network can impersonate any party, which can be exploited in many ways. However, the authors supposed in the original paper a ``reliable broadcast channel, i.e. the adversary has no control of communication''~\cite{Brandt06}. Yet even under this assumption dishonest participants can impersonate other participants by submitting messages on their behalf. Additionally, this assumption is difficult to achieve in asynchronous systems~\cite{Fischer85}. In the following we consider an attacker in control of the network, however many attacks can also be executed analogously by dishonest parties (which are considered in the original paper) in the reliable broadcast setting.

\subsection{Another attack on privacy}

Our first attack on privacy only works in the case of malleable
interactive proofs. If we switch to non-interactive non-malleable
proofs, Mallory cannot ask the other bidders for proofs using a
challenge of his choice.

However, even with non-interactive non-malleable zero-knowledge
proofs, the protocol is still vulnerable to attacks on a targeted
bidder's privacy if an attacker can impersonate any bidder of his
choice as well as the seller, which is the case for an attacker
controlling the network due to the lack of authentication. In
particular, if he wants to know Alice's bid he can proceed as follows:
\begin{enumerate}
 \item Mallory impersonates all other bidders. He starts by creating
   keys on their behalf and publishes the values $y_i$ and the
   corresponding proofs for all of them.
 \item Alice also creates her secret keyshare and publishes $y_a$
   together with a proof.
 \item Alice and Mallory compute the public key $y$.
 \item Alice encrypts her bid and publishes her $\alpha_{a j}$ and
   $\beta_{b j}$ together with the proofs.
 \item Mallory publishes $\alpha_{i j} = \alpha_{a j}$ and $\beta_{i
     j} = \beta_{a j}$ for all other bidders $i$ and also copies
   Alice's proofs.
 \item Alice and Mallory execute the computations described in the
   protocol and publish $\gamma_{i j}^a$ and $\delta_{i j}^a$.
 \item They compute $\phi_{i j}^a$ and send it to the seller.
 \item The seller publishes the $\phi_{i j}^a$ and computes the
   $v_{aj}$.
\end{enumerate}
Since all submitted bids are equal, the seller (which might also be
impersonated by Mallory) will obtain Alice's bid as the winning price,
hence it is not private any more. This attack essentially simulates a
whole instance of the protocol to make Alice indirectly reveal a bid
that was intended for another, probably real auction. To counteract
this it is not sufficient for Alice to check that the other bids are
different: Mallory can produce different $\alpha_{ij}=\alpha_{aj}y^x$
together with $\beta_{ij}=\beta_{aj}g^x$ which are still correct
encryptions of Alice bids.

Note that the same attack also works if dishonest bidders collude with the seller: they simply re-submit the targeted bidders bid as their own bid.

\subsection{Attacking fairness, non-repudiation and verifiability}

The lack of authentication obviously entails that a winning bidder can
claim that he did not submit his bid, hence violating
non-repudiation (even in the case of reliable broadcast). Additionally, this also enables an attack on
fairness: an attacker in control of the network can impersonate all
bidders vis-\`a-vis the seller, submitting bids of his choice on their
behalf and hence completely controlling the winner and winning
price. This also causes another problem with verifiability: it is
impossible to verify if the bids were submitted by the registered
bidders or by somebody else. 

\subsection{Countermeasures}

The solution to these problems is simple: all the messages need to
be authenticated, e.g. using signatures or Message Authentication
Codes (MACs) based on a trust anchor, for example a Public Key Infrastructure (PKI).

\section{Conclusion}\label{sec:concl}

In this paper we analyze the protocol of Brandt~\cite{Brandt06} from
various angles. We show that the underlying computations have a weakness which can be exploited by malicious bidders to break privacy if malleable interactive zero-knowledge proofs are used. 
We also identified two problems with verifiability and proposed
solutions. Finally we showed how the lack of authentication can be
used to mount different attacks on privacy, verifiability as well as
fairness and non-repudiation. Again we suggested a solution to address
the discovered flaws.

So sum up, the following countermeasures have to be implemented:
\begin{itemize}
 \item Use of non-interactive or non-malleable zero-knowledge proofs.
 \item All messages have to be authenticated, e.g. using a Public-Key Infrastructure (PKI) and signatures.
 \item In the outcome computation step: when computing the $\gamma^a_{ij}$ and $\delta^a_{ij}$, the bidders can check if $x_{ij}=\left( \prod_{h=1}^n \prod_{d=j+1}^k \alpha_{h d} \right) \cdot \left( \prod_{d=1}^{j-1} \alpha_{i d} \right) \cdot \left( \prod_{h=1}^{i-1} \alpha_{h j} \right)$ is equal to one -- if yes, they restart the protocol using different keys and random values. If not, they continue, and check if $\prod_a \gamma^a_{ij} = 1$. If yes, they choose different random values $m_{ij}^a$ and re-compute the $\gamma^a_{ij}$ and $\delta^a_{ij}$, otherwise they continue.
 \item In the outcome decryption step: the bidders have to prove that the value $x_a$ they used to decrypt is the same $x_a$ they used to generate their public key $y_a$ in the first step.
\end{itemize}

The attacks show that properties such as authentication can be
necessary to achieve other properties which might appear to be
unrelated at first sight, like for instance privacy. It also points
out that there is a difference between computing the winner in a fully
private way, and ensuring privacy for the bidders: in the second
attack we use modified inputs to break privacy even though the
computations themselves are secure. Additionally our analysis
highlights that the choice of interactive or non-interactive,
malleable or non-malleable proofs is an important decision in any protocol
design. 
 
As for possible generalizations of our attacks, of course the linear algebra part of our first attack is specific to this protocol. Yet the man-in-the-middle attack on malleable proofs as well as the need of authentication for privacy are applicable to any protocol. Similarly, checking all exceptional cases and ensuring that the same keys are used all along the process are also  valid insights for other protocols.

\subsubsection{Acknowledgments}
This work was partly supported by the ANR projects ProSe (decision ANR-2010-VERS-004-01) and HPAC (ANR-11-BS02-013).
We thank Dorian Arnaud, Jean-Baptiste Gheeraert, Maud Lefevre, Simon
Moura and J\'er\'emy Pouzet for spotting an error in the description
of the efficient algorithm of the attack, in an earlier version of
this paper.

\bibliography{references}

\begin{thebibliography}{10}

\bibitem{Bangerter05}
Endre Bangerter, Jan Camenisch, and Ueli~M. Maurer.
\newblock Efficient proofs of knowledge of discrete logarithms and
  representations in groups with hidden order.
\newblock In {\em Proceedings of the 8th international conference on Theory and
  Practice in Public Key Cryptography}, PKC'05, pages 154--171, Berlin,
  Heidelberg, 2005. Springer-Verlag.

\bibitem{Brandt02}
Felix Brandt.
\newblock A verifiable, bidder-resolved auction protocol.
\newblock In R.~Falcone, S.~Barber, L.~Korba, and M.~Singh, editors, {\em
  Proceedings of the 5th AAMAS Workshop on Deception, Fraud and Trust in Agent
  Societies}, pages 18--25, 2002.

\bibitem{Brandt03}
Felix Brandt.
\newblock Fully private auctions in a constant number of rounds.
\newblock In {\em Financial Cryptography 2003}, volume 2742 of {\em LNCS},
  pages 223--238. Springer, 2003.

\bibitem{Brandt06}
Felix Brandt.
\newblock How to obtain full privacy in auctions.
\newblock {\em International Journal of Information Security}, 5:201--216,
  2006.

\bibitem{Burmester89}
Mike Burmester, Yvo Desmedt, Fred Piper, and Michael Walker.
\newblock A general zero-knowledge scheme.
\newblock In {\em Advances in Cryptology - EUROCRYPT '89, Houthalen, Belgium},
  volume 434 of {\em LNCS}, pages 122--133. Springer, April 1989.

\bibitem{Chaum86}
D.~Chaum, J.H. Evertse, J.~van~de Graaf, and R.~Peralta.
\newblock Demonstrating possession of a discrete logarithm without revealing
  it.
\newblock In {\em CRYPTO'86}, volume 263 of {\em LNCS}, pages 200--212.
  Springer, 1986.

\bibitem{Chaum87}
David Chaum, Jan-Hendrik Evertse, and Jeroen van~de Graaf.
\newblock An improved protocol for demonstrating possession of discrete
  logarithms and some generalizations.
\newblock In {\em Advances in Cryptology - EUROCRYPT '87, Amsterdam, The
  Netherlands, April 13-15, 1987}, volume 304 of {\em LNCS}, pages 127--141,
  1987.

\bibitem{Chaum92}
David Chaum and Torben~P. Pedersen.
\newblock Wallet databases with observers.
\newblock In {\em Crypto'92, California, USA}, volume 0740 of {\em LNCS}, pages
  89--105. Springer, 1992.

\bibitem{Chow10}
Sherman S.~M. Chow, Changshe Ma, and Jian Weng.
\newblock Zero-knowledge argument for simultaneous discrete logarithms.
\newblock In My~T. Thai and Sartaj Sahni, editors, {\em COCOON}, volume 6196 of
  {\em LNCS}, pages 520--529. Springer, 2010.

\bibitem{Cramer98}
Ronald Cramer and Ivan Damg{\aa}rd.
\newblock Zero-knowledge proofs for finite field arithmetic, or: Can
  zero-knowledge be for free?
\newblock In Hugo Krawczyk, editor, {\em Advances in Cryptology — CRYPTO
  '98}, volume 1462 of {\em LNCS}, pages 424--441. Springer Berlin Heidelberg,
  1998.

\bibitem{Curtis07}
Brian Curtis, Josef Pieprzyk, and Jan Seruga.
\newblock An efficient {eAuction} protocol.
\newblock In {\em ARES}, pages 417--421. IEEE Computer Society, 2007.

\bibitem{ElGamal85}
Taher El~Gamal.
\newblock A public key cryptosystem and a signature scheme based on discrete
  logarithms.
\newblock In {\em Proceedings of CRYPTO 84 on Advances in cryptology}, pages
  10--18, New York, NY, USA, 1985. Springer-Verlag New York, Inc.

\bibitem{Fiat86}
Amos Fiat and Adi Shamir.
\newblock How to prove yourself: Practical solutions to identification and
  signature problems.
\newblock In Andrew~M. Odlyzko, editor, {\em CRYPTO}, volume 263 of {\em LNCS},
  pages 186--194. Springer, 1986.

\bibitem{Fischer85}
Michael~J. Fischer, Nancy~A. Lynch, and Mike Paterson.
\newblock Impossibility of distributed consensus with one faulty process.
\newblock {\em J. ACM}, 32(2):374--382, 1985.

\bibitem{Fischlin2009}
Marc Fischlin and Roger Fischlin.
\newblock Efficient non-malleable commitment schemes.
\newblock {\em Journal of Cryptology}, 22:530--571, 2009.

\bibitem{Katz2002}
Jonathan Katz.
\newblock {\em Efficient cryptographic protocols preventing "man-in-the-middle"
  attacks}.
\newblock PhD thesis, Columbia University, 2002.

\bibitem{Krishna02}
Vijay Krishna.
\newblock {\em Auction Theory}.
\newblock Academic Press, San Diego, USA, 2002.

\bibitem{Maurer09}
Ueli Maurer.
\newblock Unifying zero-knowledge proofs of knowledge.
\newblock In Bart Preneel, editor, {\em Progress in Cryptology – AFRICACRYPT
  2009}, volume 5580 of {\em LNCS}, pages 272--286. Springer Berlin Heidelberg,
  2009.

\bibitem{Naor99}
Moni Naor, Benny Pinkas, and Reuban Sumner.
\newblock Privacy preserving auctions and mechanism design.
\newblock In {\em ACM Conference on Electronic Commerce}, pages 129--139, 1999.

\bibitem{Omote01}
Kazumasa Omote and Atsuko Miyaji.
\newblock A practical {English} auction with one-time registration.
\newblock In Vijay Varadharajan and Yi~Mu, editors, {\em ACISP}, volume 2119 of
  {\em LNCS}, pages 221--234, 2001.

\bibitem{Peng02}
Kun Peng, Colin Boyd, Ed~Dawson, and Kapali Viswanathan.
\newblock Robust, privacy protecting and publicly verifiable sealed-bid
  auction.
\newblock In Robert~H. Deng, Sihan Qing, Feng Bao, and Jianying Zhou, editors,
  {\em ICICS}, volume 2513 of {\em LNCS}, pages 147--159. Springer, 2002.

\bibitem{Sadeghi02}
Ahmad-Reza Sadeghi, Matthias Schunter, and Sandra Steinbrecher.
\newblock Private auctions with multiple rounds and multiple items.
\newblock In {\em DEXA Workshops}, pages 423--427. IEEE, 2002.

\bibitem{Sako00}
Kazue Sako.
\newblock An auction protocol which hides bids of losers.
\newblock In Hideki Imai and Yuliang Zheng, editors, {\em Public Key
  Cryptography}, volume 1751 of {\em LNCS}, pages 422--432. Springer, 2000.

\bibitem{Schnorr91}
C.~P. Schnorr.
\newblock Efficient signature generation by smart cards.
\newblock {\em Journal of Cryptology}, 4:161--174, 1991.

\end{thebibliography}
\bibliographystyle{plain}
\end{document}